%% file: main.tex
\documentclass[colorlinks=true,urlcolor=blue,citecolor=blue,linkcolor=red,UKenglish,runningheads,numbers=enddot]{llncs}
\usepackage{microtype}

\usepackage{amssymb,amsmath,amsthm}

\newtheorem{observation}{{\bf Observation}}

\usepackage{wrapfig}
\usepackage{xspace}
\usepackage{xcolor}
\usepackage{paralist}
\usepackage{hyperref}
\hypersetup{pdfborder={0 0 0.4}}
\usepackage{todonotes}
\usepackage{graphicx}
\usepackage[font=small,labelfont=bf]{caption}
\usepackage{subfig}
\usepackage{thm-restate}
\usepackage{cleveref}
\usepackage{algpseudocode,algorithm}
\usepackage{geometry}
\newgeometry{left=2cm,right=2cm,top=2cm,bottom=2cm} 

\renewenvironment{proof}{\noindent{\bf Proof.~}}{\hspace*{\fill}${\small\square}$\vspace{2mm}}

\renewcommand{\paragraph}[1]{\smallskip\noindent{\bf #1}\xspace}
\newcommand{\remove}[1]{}

\def\etal{\emph{et~al.}\xspace}

\graphicspath{{figures/}}

\definecolor{blue}{rgb}{0.274,0.392,0.666}
\definecolor{red}{rgb}{1,0.3,0.3}
\definecolor{green}{rgb}{0,0.588,0.509}

\newcommand{\red}[1]{{\color{red}{#1\xspace}}}

\newcommand{\calC}[1]{\ensuremath{{\cal C}^{#1}}\xspace}
\newcommand{\calT}[1]{\ensuremath{{\cal T}^{#1}}\xspace}

\DeclareMathOperator{\emw}{emw}
\DeclareMathOperator{\tw}{tw}

\DeclareMathOperator{\msotwo}{MSO_2}

\DeclareMathOperator{\planar}{\textsc{planar}}
\DeclareMathOperator{\connected}{\textsc{conn}}
\DeclareMathOperator{\cplanar}{\textsc{c-planar}}

\graphicspath{{figures/}}

\newcommand{\titletext}{
Subexponential-Time and FPT Algorithms for\\ Embedded Flat Clustered Planarity}
\newcommand{\shorttitle}{Subexponential-Time and FPT Algorithms for C-Planarity Testing}
\titlerunning{\shorttitle\xspace}

\authorrunning{{Da Lozzo}~\etal}
\title{\titletext}
\author{Giordano {Da Lozzo}\inst{1} \and David Eppstein\inst{2} \and  Michael T. Goodrich\inst{2} \and  
Siddharth Gupta\inst{2}}
\institute{
Department of Computer Science, Roma Tre University, Rome, IT
\email{dalozzo@dia.uniroma3.it}
\and
Department of Computer Science, University of California, Irvine, USA
\\
\email{\{eppstein,goodrich,guptasid\}@uci.edu}
}

\begin{document}
\maketitle


\begin{abstract}
The {\sc C-Planarity} problem asks for a drawing of a \emph{\mbox{clustered} graph}, i.e., a
graph whose vertices belong to properly nested clusters, in which each cluster is represented by a simple closed region with no edge-edge 
crossings, no region-region crossings, and no unnecessary edge-region crossings. 
We study \mbox{\sc C-Planarity} for \emph{embedded flat clustered graphs}, 
graphs with a fixed combinatorial embedding whose clusters partition the vertex set.
Our main result is a subexponential-time algorithm to test \mbox{\sc C-Planarity} for 
these graphs when their face size is bounded.
Furthermore, we consider a variation of the notion of {\em embedded tree decomposition} in which,
for each face, including the outer face, there is a bag that contains every vertex
of the face. We show that \mbox{\sc C-Planarity} is fixed-parameter tractable with
the embedded-width of the underlying graph and the number of disconnected clusters as parameters.
\end{abstract}

\section{Introduction}\label{se:intro}
\begin{wrapfigure}[9]{R}{.31\textwidth}
    \centering
    \vspace{-11mm}
    \includegraphics[page=1, width=.31\textwidth]{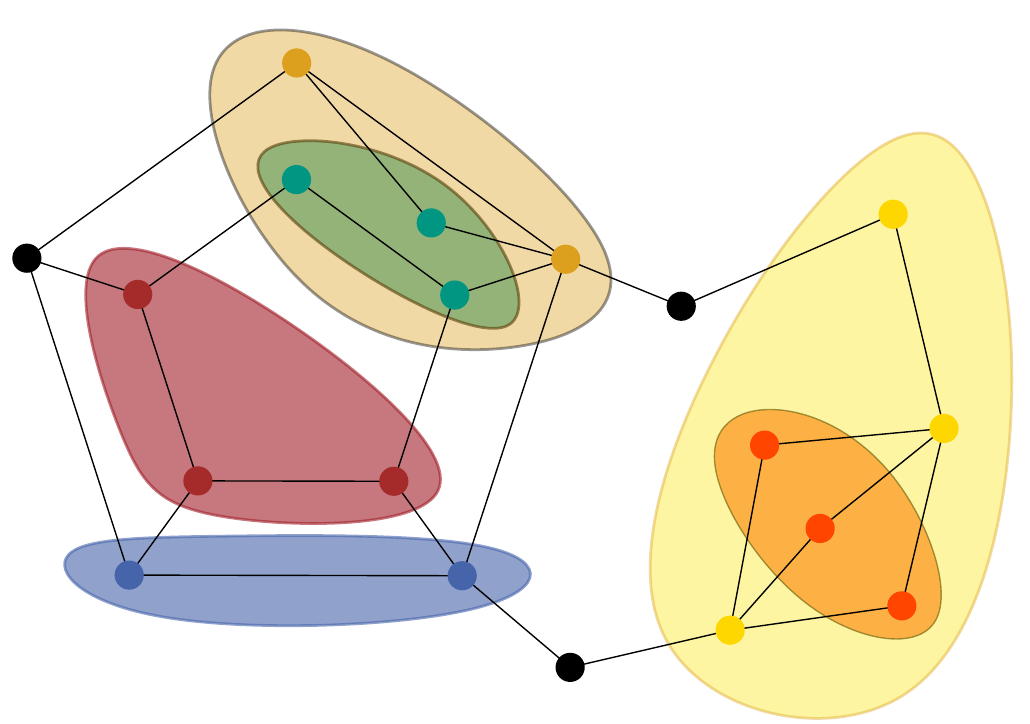}
    \vspace{-6mm}
    \caption{A c-planar drawing}\label{fig:c-planar}
\end{wrapfigure}
A \emph{clustered graph} (or \emph{c-graph}) is a pair $\calC{}(G,\calT{})$ with \emph{underlying graph} $G$ and \emph{inclusion tree} $\calT{}$, i.e., a rooted tree whose leaves are the vertices of $G$. Each internal node $\mu$ of $\calT{}$ represents a cluster of vertices of~$G$ (its leaf descendants) which induces a subgraph $G(\mu)$.
A~\emph{c-planar drawing} of $\calC{}(G,\calT{})$ (Fig.~\ref{fig:c-planar}) consists of a drawing of $G$ and of a representation of each cluster $\mu$ as a {\em simple closed region} $R(\mu)$, i.e., a region homeomorphic to a closed disc, such that:
\begin{inparaenum}[(1)]
\item Each region $R(\mu)$ contains the drawing of $G(\mu)$.
\item  For every two clusters $\mu,\nu \in \calT{}$, \mbox{$R(\nu) \subseteq R(\mu)$} if and only if $\nu$ is a descendant of $\mu$ in $\calT{}$.
\item No two edges cross.
\item No edge crosses any region boundary more than once.
\item No two region boundaries intersect.
\end{inparaenum}


An interesting and challenging line of research in graph drawing concerns the computational complexity of the \mbox{\sc C-Planarity} problem, which asks to test the existence of a c-planar drawing of a c-graph.
This problem is notoriously difficult, particularly when (as in Fig.~\ref{fig:c-planar}) clusters may be disconnected, faces may have unbounded size,
and the cluster hierarchy may have multiple nested levels.
No known subexponential-time algorithm solves the (general) \mbox{\sc C-Planarity} problem, and
it is unknown whether it is NP-complete, although the related problem of 
 splitting as few clusters as possible to make a c-graph c-planar was proved NP-hard~\cite{Angelini2010}.
Thus, there is considerable interest in 
subexponential-time, slice-wise polynomial, and fixed-parameter tractable algorithms, besides polynomial-time algorithms for special cases
of {\sc C-Planarity}.


{\sc C-Planarity} was introduced
by Feng, Cohen, and Eades~\cite{fce-pcg-95},
who solved it in quadratic time for the \emph{c-connected} case
when every cluster induces a connected subgraph.
Similar  results were given by 
Lengauer~\cite{Lengauer:1989} using different terminology.
Dahlhaus~\cite{Dahlhaus1998} claimed a linear-time algorithm for c-connected
{\sc C-Planarity} (with some details later provided by
Cortese \etal~\cite{cdfpp-cccg-06}).
\mbox{Goodrich \etal~\cite{Goodrich2006}} gave a cubic-time algorithm for  disconnected clusters that 
satisfy an ``extroverted'' property,
and Gutwenger {\it  et al.}~\cite{Gutwenger2002} provided a polynomial-time
algorithm for ``almost'' c-connected inputs.
Cornelsen and Wagner showed polynomiality for completely connected c-graphs, i.e., c-graphs for which
not only every cluster but also the complement of each cluster is connected~\cite{CornelsenW06}.
FPT algorithms have also been investigated~\cite{BlasiusR16,ck-ssscp-12}. For additional special cases
, see, e.g.,~\cite{ad-cpwp-16,addf-sprepg-17,AngeliniLBFPR15,AthenstadtC17,cgjkm-cmcps-08,DidimoGL08}. 

A c-graph is \emph{flat} when no non-trivial cluster is a subset of another, so $\calT{}$ has only three levels: the root, the clusters, and the leaves.
Flat  {\sc C-Planarity} can be solved 
in polynomial time
for embedded c-graphs with at most $5$ vertices per face~\cite{df-ectefgsf-13,FulekKMP15}
or at most two vertices of each cluster per face~\cite{cdfk-atcpefcg-14}, 
for embedded c-graphs in which each cluster induces a subgraph with at most two connected components~\cite{JelinekJKL08},
and
for c-graphs with two clusters~\cite{b-dppIIItcep-98,FulekKMP15,hn-sat2pepg-14} or three clusters~\cite{Akitaya17}.
At the other end of the size spectrum, 
Jel{\'\i}nkov{\'a} \etal~\cite{Jelinkova2008} provide efficient
algorithms for $3$-connected flat c-graphs when
each cluster has at most $3$ vertices.
Fulek~\cite{Fulek16} 
speculates that {\sc C-Planarity} could be
solvable in subexponential time for more general embedded flat c-graphs.

\clearpage
\paragraph{New Results.}
In this paper, we provide subexponential-time and 
fixed-parameter tractable algorithms for broad classes of c-graphs.
\mbox{We show the following results:}

\begin{itemize}[$\diamond$]
\item{\sc C-Planarity} can 
be solved in subexponential time for embedded flat \mbox{c-graphs} with bounded face size ({Section}~\ref{se:exponential}).
\item{\sc C-Planarity} is fixed-parameter tractable for embedded flat c-graphs with embedded-width and number of disconnected clusters as parameters \mbox{({Section}~\ref{se:courcelle}).} 
\end{itemize}

Our first result uses divide-and-conquer with a large but subexponential branching factor. It
exploits cycle separators in planar graphs and a concise representation 
of the connectivity of each cluster in a c-planar drawing.
This method also leads to an $\mathsf{XP}$ algorithm for \emph{generalized $h$-simply nested graphs}, which extend simply-nested graphs with bounded face size (Section~\ref{se:simplynested}). Recall that, $\mathsf{XP}$ (short for \emph{slice-wise polynomial}) is the class of parameterized problems with input size $n$ and parameter $k$ than can be solved in $O(n^{f(k)})$ time, where $f$ is a computable function.

We obtain our second result by expressing 
c-planarity in extended monadic second-order 
logic for embedded flat c-graphs and applying Courcelle's~Theorem.
The graphs to which this result applies, with bounded treewidth and bounded face size, include 
the \emph{nested triangles graphs}, a standard family of examples that are hard for many graph drawing tasks, 
the {\em dual graphs of bounded-treewidth bounded-degree plane graphs}~\cite{BouchitteMT01}, and the \emph{buckytubes},
 graphs formed from a planar hexagonal lattice wrapped to form a cylinder of bounded diameter. 
 

\input{preliminaries}
\input{subexponential}

\input{courcelle}

\section{Conclusions and Open Problems}\label{se:conclusions}

In this paper, we provide subexponential-time, $\mathsf{XP}$, and FPT \mbox{algorithms} to test {\sc C-Planarity} of fairly-broad classes of 
c-graphs.

\setcounter{footnote}{2}
Several interesting questions arise from this research:
\begin{inparaenum}[(1)]
	\item Can our results be generalized from flat to non-flat 
	c-graphs?
	\item Is there a fully polynomial-time algorithm to test {\sc C-Planarity} of c-graphs whose underlying graph is a generalized $h$-simply-nested graph? 
	\item Are there interesting parameters of the underlying graph such that {\sc C-Pla\-na\-ri\-ty} is FPT with respect to a single one of them (e.g., outerplanarity index, maximum face size, notable graph width parameters)?
	\item Are there interesting parameters of the c-graph such that \mbox{\sc C-Planarity} is FPT with respect to a single one of them (e.g., number of clusters, number of vertices of the same cluster incident to the same face\footnote{This question has also been previously asked by Chimani {\em et al.}~\cite{cdfk-atcpefcg-14}}, maximum distance between two faces containing vertices of the same cluster)?
\end{inparaenum}
%
\clearpage{}
\bibliographystyle{splncs03}
\bibliography{bibliography}


%
%

\end{document}

%% file: preliminaries.tex
\section{Definitions and Preliminaries}\label{se:prel}

The graphs considered in this 
paper are finite, simple, and connected. A graph is \emph{planar} if it admits a drawing in the plane without edge crossings.
A \emph{combinatorial embedding} is an equivalence class of planar drawings, where two drawings of a graph are \emph{equivalent} if they determine the same \emph{rotation} at each vertex, i.e, the same circular orderings for the edges around each vertex.
A planar drawing partitions the plane into topologically connected regions, called \emph{faces}. 
The bounded faces are the \emph{inner faces}, while
the unbounded face is the \emph{outer face}.
A combinatorial embedding together with a choice for the outer face defines a \emph{planar embedding}.
An \emph{embedded graph} (\emph{plane graph}) is a planar graph with a fixed combinatorial embedding (fixed planar embedding).
%
The \emph{length} of a face $f$ is the number of occurrences of edges encountered in a traversal of $f$.
The \emph{maximum face size} of an embedded graph is the length of its largest face.

A graph is \emph{connected} if it contains a path between any
two vertices.  A \emph{cut-vertex} is a vertex whose removal
disconnects the graph.  A \emph{separation pair} is a pair of
vertices whose removal disconnects the graph.  A connected graph is
\emph{$2$-connected} if it contains at least $3$ vertices and does not have a cut-vertex, and a $2$-connected
graph is \emph{$3$-connected} if it contains at least $4$ vertices and does not have a separation pair. The \emph{blocks} of a graph are its maximal $2$-connected subgraphs. Any (subdivision of a) $3$-connected planar graph admits a unique combinatorial embedding (up to a flip)~\cite{Whitney33}.

\paragraph{Tree-width and Embedded-width.}
A \emph{tree decomposition} of a graph $G$ is a tree $T$ whose nodes, called \emph{bags}, are labeled by subsets of vertices of $G$. For each vertex $v$ the bags containing $v$ must form a nonempty contiguous subtree of $T$, and for each edge $uv$ at least one bag must contain both $u$ and $v$. The \emph{width} of the decomposition is one less than the maximum cardinality of any bag, and the \emph{treewidth} $\tw(G)$ of $G$ is the minimum width of any of its tree decompositions. 

Recently, Borradaile \etal~\cite{BELW} developed a variant of treewidth, specialized for plane graphs, called embedded-width. According to their definitions, a tree decomposition \emph{respects} an embedding of a plane graph $G$ if, for every inner face $f$ of $G$, at least one bag contains all the vertices of $f$. They define the \emph{embedded-width} $\emw(G)$ of $G$ to be the minimum width of a tree decomposition that respects the embedding of $G$. 
%
%
We will use the following result~\cite{BELW}.

\begin{theorem}[\cite{BELW}, Theorem~2]\label{th:embedded}
If $G$ is a plane graph where every inner face has length at most $\ell$, then
\mbox{$\emw(G) \leq (\tw(G) + 2) \cdot \ell - 1$}.
\end{theorem}

Borradaile \etal do not require the vertices of the outer face to be contained in a same bag.
In our applications, we modify this concept so that the tree decomposition also includes a bag containing the outer face, and we denote the minimum width of such a tree decomposition as $\overline{\emw}(G)$.  We have the following.

\begin{lemma}\label{le:newembedded}
If $G$ is a plane graph whose maximum face size (including the size of the outer face) is $\ell$, then
$\overline{\emw}(G) = O(\ell \cdot \tw(G))$.
\end{lemma}

\begin{proof}
To prove the statement, we can proceed as follows. 

We augment $G$ to a graph $G'$, 
by embedding $G$ in the interior of a triangle~$\Delta$ and by {identifying} one of the vertices of the outer face of $G$ with a vertex of~$\Delta$. 
Clearly, \mbox{$tw(G') = \max(\tw(G),2)$} and $G'$ has maximum face size $\ell'=O(\ell)$. 
Also, we have that $\overline{\emw}(G) \leq {\emw}(G')$, since $G \subseteq G'$ and since all the vertices incident to the same face in $G$ are also incident to the same face in $G'$.
Thus, the statement follows from the fact that, by Theorem~\ref{th:embedded}, \mbox{$\emw(G') \leq (\tw(G') + 2) \cdot \ell' - 1$}.\end{proof}

\paragraph{Clustered Planarity.}
Recall that, in a c-graph $\calC{}(G,\calT{})$, each internal node $\mu$ of $\calT{}$ corresponds to the set $V(\mu)$ of vertices of $G$ at leaves of the subtree of $\calT{}$ rooted at $\mu$.
Set $V(\mu)$ induces the subgraph $G(\mu)$ of $G$. We call the internal nodes
other than the root \emph{clusters}.
A cluster $\mu$ is \emph{connected} if $G(\mu)$ is connected and \emph{disconnected} otherwise.
A c-graph $\calC{}(G,\calT{})$ is \emph{c-connected} if every cluster is connected.

A c-graph is \emph{c-planar} if it admits a c-planar drawing. Two c-graphs $\calC{}(G,\calT{})$ and $\calC{'}(G',\calT{'})$ are \emph{equivalent} if both are c-planar or neither is.
If the root of $\calT{}$ has leaf children, enclosing each leaf $v$ in a new singleton cluster produces an equivalent c-graph. Therefore, we can safely assume that each vertex belongs to a cluster. 
A c-graph is \emph{flat} if each leaf-to-root path in $\calT{}$ has exactly three nodes. The clusters of a flat c-graph form a partition of the vertex set.

An \emph{embedded c-graph} $\calC{}(G,\calT{})$ is a c-graph whose underlying graph has a fixed combinatorial embedding.  It is \emph{c-planar} if it admits a c-planar drawing that preserves the embedding of $G$. In what follows, we only deal with embedded flat c-graphs. Therefore, we will refer to such graphs simply as c-graphs.

We define the \emph{candidate saturating edges} of a c-graph $\calC{}(G,\calT{})$ as follows. 
For each face $f$ of $G$, let $G(f)$ be the closed walk composed of the vertices and edges of~$f$. 
For each cluster $\mu \in \calT{}$, consider the set ${\cal X}_\mu(f)$ of connected components of $G(f)$ induced by the vertices of $\mu$ and, for each component $\xi \in {\cal X}_\mu(f)$, assign a vertex of~$f$ belonging to $\xi$ as a reference vertex of $\xi$. 
We add an edge inside~$f$ between the reference vertices of any two components in ${\cal X}_\mu(f)$ if and only if such vertices belong to different connected components of $G(\mu)$; see Figs.~\ref{fig:candidates-a} and~\ref{fig:candidates-b}.
\begin{figure}[tb]
    \centering
  \subfloat{
    \includegraphics[page=1,width=.3\textwidth]{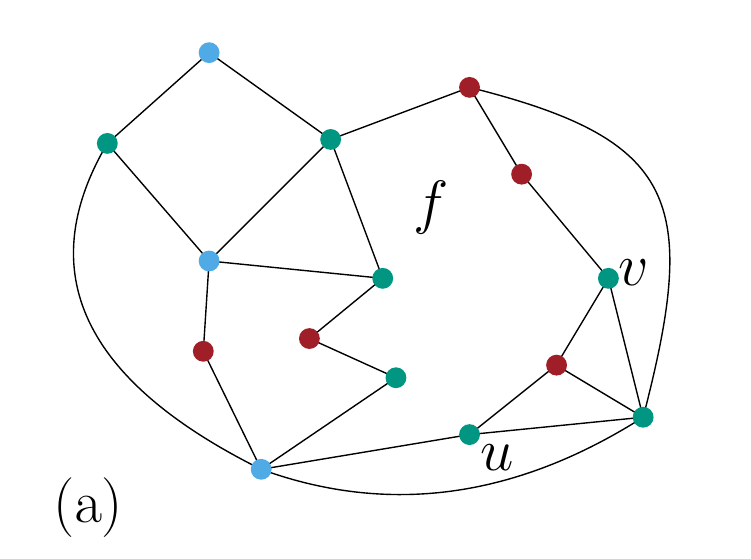}
    \label{fig:candidates-a}
    }
  \subfloat{
    \includegraphics[page=2,width=.3\textwidth]{saturatingEdges}
    \label{fig:candidates-b}
    }
  \hfil
  \subfloat{
    \includegraphics[page=3,width=.3\textwidth]{saturatingEdges}
    \label{fig:candidates-c}
  }
  \caption{(a) An embedded flat c-graph $\calC{}(G,\calT{})$. 
  (b) A super c-graph of $\calC{}$ containing all the candidate saturating edges of $\calC{}$ (thick and colored curves); since vertices $u$ and $v$ belong to different components of $X_\mu(f)$ but to the same connected component of $G(\mu)$,  edge $(u,v)$ is not a candidate saturating edge.
  (c) A super c-graph of $\calC{}$ satisfying Condition~(\ref{con:3}) of Theorem~\ref{th:characterization}; regions enclosing vertices of each cluster are shaded.
  }
  \label{fig:candidates}
\end{figure}
%
A c-graph obtained from $\calC{}(G,\calT{})$ by adding to $\calC{}$ a subset $E^+$ of its candidate saturating edges is a \emph{super c-graph} of~$\calC{}$.

Di~Battista and Frati~\cite{df-ectefgsf-13} gave the following characterization.

\begin{theorem}[\cite{df-ectefgsf-13}, Theorem~1]\label{th:characterization}
A c-graph $\calC{}(G,\calT{})$ is c-planar if and only if: 
\begin{enumerate}[(i)]
\item\label{con:1} $G$ is planar; 
\item\label{con:2} there exists a face $f$ in $G$ such that when $f$ is chosen as the outer face for $G$ no cycle composed of vertices of the same cluster encloses a vertex of a different cluster in its interior; and
\item\label{con:3} 
there exists a super c-graph $\calC{'}(G',\calT{})$ of $\calC{}$ such that $G'$ is planar and $\calC{'}$ is c-connected (see Fig.~\ref{fig:candidates-c}).
\end{enumerate}
\end{theorem}

Conditions~(\ref{con:1}) and~(\ref{con:2}) of Theorem~\ref{th:characterization} can be easily verified in linear time. Therefore, we can assume that any c-graph satisfies these conditions. Following~\cite{df-ectefgsf-13} we thus view the problem of testing c-planarity as one of testing Condition~(\ref{con:3}).

A \emph{cluster-separator} in a c-graph  $\calC{}(G,\calT{})$ is a cycle $\rho$ in $G$ for which some cluster $\mu \in \calT{}$ has vertices both in the interior and in the exterior of $\rho$ but with $V(\mu) \cap V(\rho) = \emptyset$. 
Condition~(\ref{con:3}) immediately yields the following observation. 

\begin{observation}\label{obs:safe-cycle} A c-graph that has a cluster-separator
is not c-planar.
\end{observation}

In the next sections, it will be useful to only consider c-graphs which are at least $2$-connected (Section~\ref{se:exponential}) and $3$-connected (Section~\ref{se:courcelle}). 
The next lemma, conveniently stated in a stronger form\footnote{In Section~\ref{se:courcelle}, we exploit all the properties of the lemma. In Section~\ref{se:exponential}, we only exploit the existence of an equivalent $2$-connected c-graph with maximum face size~$\kappa=O(\ell)$.}, shows that this is not a loss of generality.



\begin{restatable}{lemma}{lemmathreeconnected}\label{lem:3connected}
Let $\calC{}(G,\calT{})$ be an $n$-vertex c-graph with maximum face size $\ell$. There exists an $O(n)$-time algorithm that constructs an equivalent c-graph $\calC{*}(G^*,\calT{*})$ with $|V(G^*)|=O(n){}$ such that:
\begin{inparaenum}
 \item  $G^*$ is $3$-connected,
 \item  the maximum face size $\kappa$ of $G^*$ is $O(\ell)$, and
 \item  the 
 c-graph $\calC{\diamond}(G^\diamond,\calT{\diamond})$ obtained by augmenting $\calC{*}(G^*,\calT{*})$ with all its candidate saturating edges is such that $\tw(G^\diamond) = O(\overline{\emw}(G))$.
 \end{inparaenum}
\end{restatable}
\begin{proof}
\begin{figure}[tb]
     \centering
     \subfloat[\label{fi:t2}]{
     \includegraphics[page=2, width=.31\textwidth]{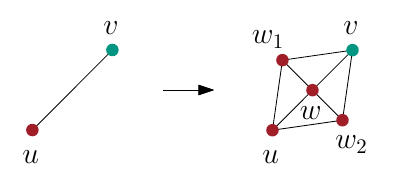}
     }\hfil
     \subfloat[\label{fi:t3}]{
     \includegraphics[page=3, width=.31\textwidth]{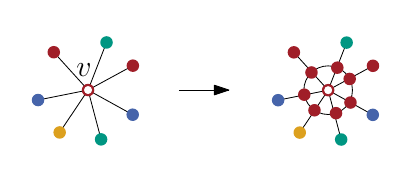}
     }
     \caption{Transformations for the proof of Lemma~\ref{lem:3connected}.}\label{fig:transformations}
\end{figure}
To prove the statement, we can proceed as follows. 


First, we transform c-graph $\calC{}(G,\calT{})$ into an equivalent c-graph $\calC{'}(G',\calT{'})$,
by applying the transformation in Fig.~\ref{fi:t2} to every edge,
such that $|V(G')|=O(|V(G)|)$, every vertex of $G'$ has degree at least~$3$, the maximum face size $\ell'$ of $G'$ is $O(\ell)$, $tw(G') = \tw(G)$,
and each vertex $u$ of $G'$ is incident to at least three edges in each block $u$ belongs to.

Second, we transform c-graph $\calC{'}(G',\calT{'})$ into an equivalent c-graph $\calC{*}(G^*,\calT{*})$, by applying the transformation in Fig.~\ref{fi:t3} to every vertex,
such that $|V(G^*)|=O(|V(G')|)$, $G^*$ is $3$-connected, the maximum face size $\ell^*$ of $G^*$ is $O(\ell')$, and the c-graph $\calC{\diamond}(G^\diamond,\calT{\diamond})$ obtained by augmenting $\calC{*}(G^*,\calT{*})$ with all its candidate saturating edges is such that $\tw(G^\diamond) = O(\ell \cdot \tw(G))$, which implies that 
$\tw(G^\diamond) = \overline{\emw}(G)$, since $\overline{\emw}(G) = O(\ell \cdot \tw(G))$ by Lemma~\ref{le:newembedded}.

We now describe each of the transformations in detail.


First, initialize $\calC{'}=\calC{}$. For every vertex $c$ of $G$, let $(c,x)$ be any edge incident to $c$.
Add to $G'$ vertices $w_1$ and $w_2$ and embed paths $(c,w_1,x)$ and $(c,w_2,x)$ in the interior of each of the two faces of $G'$ edge $(c,x)$ is incident to; also, subdivide edge $(c,x)$ with a vertex $w$, add edges $(w_1,w)$ and $(w_2,w)$, and assign vertices $w_1$, $w_2$, and $w$ to the same cluster of $\calT{'}$ ($\calT{}$) vertex $c$ belongs to. Refer to Fig.~\ref{fi:t2}.
By construction, all the newly added vertices have degree at least $3$. In particular, observe that each cut-vertex of $G'$ is incident to at least three edges in each of the blocks such a cut-vertex belongs to.
It is easy to see that $\calC{'}$ and $\calC{}$ are equivalent. Also, the maximum face size $\ell'$ of $G'$ is $O(\ell)$. 
Further, $\tw(G')=\max(\tw(G),3)$, as the transformation replaces edges with subgraphs of treewidth $3$.

Second, initialize $\calC{*}=\calC{'}$. For every vertex $c$ of $G'$, we subdivide each edge $(c,x_i)$ incident to $c$ with a dummy vertex $v_i$. Denote such a graph by $G^+$. 
Also, add an edge between any two vertices $v_i$ and $v_j$ such that edges $(c,x_i)$ and $(c,x_j)$ are consecutive around $c$ in the unique face shared by $v_i$ and $v_j$ in $G^+$. Finally, assign each vertex $v_i$ to the same cluster of $\calT{*}$ ($\calT{'}$) vertex $c$ belongs to. Refer to Fig.~\ref{fi:t3}.
The equivalence between $\calC{*}(G^*,\calT{*})$ and $\calC{'}(G',\calT{'})$ is again straightforward.
Clearly, $|V(G^*)|=O(|V(G')|)$ and $\ell^*=O(\ell')$.
Also, by the observation that the cut-vertices of $G'$ are incident to at least three edges in each of the blocks such cut-vertices belong to, the applied transformation fixes the rotation at all the vertices of $G^*$. Since each vertex of $G^*$ has minimum degree $3$ and $G^*$ has a fixed combinatorial embedding (up to a flip), by the result of Whitney~\cite{Whitney33}, we have that $G^*$ is $3$-connected. 
Furthermore, $\overline{\emw}(G^*)= O(\overline{\emw}(G'))$, since $G^*$ is obtained by subdividing each edge of $G'$ twice, thus obtaining a graph $G^+$ with the same tree-width as $G'$ and maximum face-size in $O(\ell')$, and by adding edges between some of the vertices incident to the faces of $G^+$.
Since $G^\diamond$ is the graph obtained by adding all the candidate saturating-edges of $G^*$ (recall that such edges only connect vertices in the same face of $G^*$), we have that $\tw(G^\diamond)=O(\overline{\emw}(G^*))=O(\overline{\emw}(G'))$. 
Since, by Lemma~\ref{le:newembedded}, $\overline{\emw}(G') = O(\ell' \cdot \tw(G'))$ and since $\ell' = O(\ell)$ and $\tw(G') = O(\tw(G))$, we have
that $\tw(G^\diamond)= O(\ell \cdot \tw(G))$.
This concludes the proof of the lemma.\xspace
\end{proof}

%% file: subexponential.tex
\section{A Subexponential-Time Algorithm for C-Planarity}\label{se:exponential}

In this section, we describe a divide-and-conquer algorithm for testing the \mbox{c-planarity} of $2$-connected c-graphs exploiting cycle separators in planar graphs. 


The ``conquer'' part of our divide-and-conquer uses the following operation on pairs of c-graphs. 
Let $G_1$ and $G_2$ be plane graphs on overlapping vertex sets such that the outer face of $G_1$ and an inner face of $G_2$ are bounded by the same cycle $\rho$.
\emph{Merging} $G_1$ and $G_2$ constructs a new plane graph $G$  from $G_1 \cup G_2$ as follows. We remove multi-edges (belonging to cycle $\rho$)  and assign each vertex $v$ a rotation whose restriction to the edges of $G_2$ (of $G_1$) is the same as the rotation at $v$ in $G_2$ (in $G_1$). This is possible as cycle $\rho$ bounds the outer face of $G_1$ and an inner face of $G_2$. We say that $G$ is a \emph{merge} of $G_1$ and $G_2$.
Now consider two c-graphs $\calC{}_1(G_1,\calT{}_1)$ and $\calC{}_2(G_2,\calT{}_2)$ such that (i) $G_1 \cap G_2 = \rho$ is a cycle, (ii) for each vertex $v \in V(\rho)$, vertex $v$ belongs to the same cluster $\mu$ both in $\calT{}_1$ and in $\calT{}_2$,  and (iii) cycle $\rho$ bounds the outer face of $G_1$ and an inner face of $G_2$ (when a choice for their outer faces that satisfies Condition~(\ref{con:2}) of Theorem~\ref{th:characterization} has been made). \emph{Merging} $\calC{}_1$ and $\calC{}_2$ is the operation that constructs a c-graph $\calC{}(G,\calT{})$ as follows. Graph $G$ is obtained by merging $G_1$ and $G_2$.
Tree $\calT{}$ is obtained as follows. Initialize $\calT{}$ to $\calT{}_1$. First, for each cluster $\mu \in \calT{}_2 \cap \calT{}_1$, we add the leaves of $\mu$ in $\calT{}_2$ as children of cluster $\mu$ in $\calT{}$, removing duplicate leaves. Second, for each cluster $\mu \in \calT{}_2 \setminus \calT{}_1$, we add the subtree of $\calT{}_2$ rooted at $\mu$ as a child of the root of $\calT{}$. We say that $\calC{}(G,\calT{})$ is a \emph{merge} of $\calC{}_1(G_1,\calT{}_1)$ and $\calC{}_2(G_2,\calT{}_2)$.

In the ``divide'' part of the divide-and-conquer, we replace subgraphs of the input by smaller planar components called \emph{cycle-stars} that preserve their c-planarity properties.
Let $G$ be a connected plane graph that contains a face whose boundary is a cycle~$\rho$. We say that $G$ is a \emph{cycle-star} if removing all the edges of $\rho$ makes $G$ a forest of stars; refer to Fig.~\ref{fig:lemma-c}. Also, we say that cycle $\rho$ is \emph{universal} for~$G$ and we say that a vertex of $G$ is a \emph{star vertex} of $G$ if it does not belong to $\rho$. Clearly, the size of $G$ is $O(|\rho|)$.

For a c-planar c-graph $\calC{}(G,\calT{})$ and a cycle separator $\rho$, we denote by $\calC{+}_\rho(G^+,\calT{+})$ (by $\calC{-}_\rho(G^-,\calT{-})$) the c-graph obtained from $\calC{}$ by removing all the vertices and the edges of $G$ that lie in the interior of $\rho$ (in the exterior of $\rho$). 
%
%
%
Consider a super c-graph $\calC{'}(G',\calT{})$ of $\calC{}$ satisfying Condition~(\ref{con:3}) of Theorem~\ref{th:characterization}, which exists since $\calC{}$ is c-planar.
We now give a procedure, which will be useful throughout the paper, to construct two special c-planar c-graphs $\calC{-}(S^-,{\cal K}^-)$ and $\calC{+}(S^+,{\cal K}^+)$ associated with $\calC{'}$ whose properties \mbox{are described in the following lemma.}

\begin{restatable}{lemma}{LemmaSplit}\label{le:split}
C-graphs $\calC{-}(S^-,{\cal K}^-)$ and $\calC{+}(S^+,{\cal K}^+)$ are such that:
\begin{enumerate}
\item\label{prop:0} graph $S^-$ ($S^+$) is a cycle-star whose universal cycle is $\rho$,
\item\label{prop:1} cycle $\rho$ bounds the outer face of $S^-$ (an inner face of $S^+$), 
\item\label{prop:2} each star vertex of $S^-$ ($S^+$) and all its neighbours belong to the same cluster in ${\cal K}^-$ (${\cal K}^+$), and
\item\label{prop:3} the c-graph $\calC{}_{out}$ ($\calC{}_{in}$) obtained by merging $\calC{-}(S^-,{\cal K}^-)$ and $\calC{+}_\rho(G^+,\calT{+})$ (by merging $\calC{+}(S^+,{\cal K}^+)$ and $\calC{-}_\rho(G^-,\calT{-})$) is c-planar.
\end{enumerate}
\end{restatable}

We describe how to construct $\calC{-}(S^-,{\cal K}^-)$ from $\calC{'}$; refer to Fig.~\ref{fig:replacement}. The construction of $\calC{+}(S^+,{\cal K}^+)$ is symmetric.

First, for each cluster $\mu$ such that $V(\mu) \cap V(\rho) = \emptyset$, we remove all the vertices in $V(\mu)$ lying in the interior of $\rho$ together with their incident edges.
By Observation~\ref{obs:safe-cycle}, the resulting c-graph is still c-planar and c-connected.
Also, we remove edges in the interior of $\rho$ whose endpoints belong to different clusters. Clearly, this simplification preserve c-connectedness. \mbox{We still denote the resulting c-graph as~$\calC{'}$.}
 
Second, consider the c-graph $H$ consisting of the vertices and of the edges of $\calC{'}$ lying in the interior and along the boundary of $\rho$. For each cluster $\mu$ and for each connected component $c^i_\mu$ of $\mu$ in $H$, we replace all the vertices and edges of $c^i_\mu$ lying in the interior of $\rho$ in $\calC{'}$ with a single vertex $s^i_\mu$, assigning it to the same cluster $\mu$ and making it adjacent to all the vertices in $V(c^i_\mu) \cap V(\rho)$. Let $\calC{*}$ be the resulting c-graph. It is easy to see that such a transformation preserves c-connectedness and planarity, therefore $\calC{*}$ is a c-connected c-planar c-graph.
By construction, each vertex $v \in V(\rho)$ is adjacent to a single vertex $s^i_\mu$, where $\mu$ is the cluster vertex $v$ belongs to; thus, the vertices and the edges in the interior and along the boundary of $\rho$ in $\calC{*}$ form 
c-graph $\calC{-}(S^-,{\cal K}^-)$ whose underlying graph $S^-$ is a cycle-star satisfying Properties~(\ref{prop:0}),~(\ref{prop:1}) and~(\ref{prop:2}) of Lemma~\ref{le:split}. 
Further, since the subgraph of $\calC{*}$ consisting of the vertices and of the \mbox{edges lying in the} exterior and along the boundary of $\rho$ coincides with $\calC{+}_\rho(G^+,\calT{+})$, we have that $\calC{*}$ is a c-planar c-connected super c-graph of $\calC{}_{out}$. Thus, by Condition~(\ref{con:3}) of Theorem~\ref{th:characterization}, Property~(\ref{prop:3}) of Lemma~\ref{le:split} is also satisfied.

\begin{figure}[tb!]
    \centering
    \subfloat
    {
    \includegraphics[page=9,width=.23\textwidth]{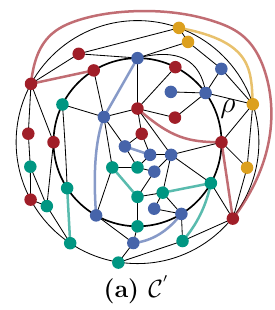}
    \label{fig:lemma-a}
    }
    \hfil
  \subfloat
  {
    \includegraphics[page=10,width=.23\textwidth]{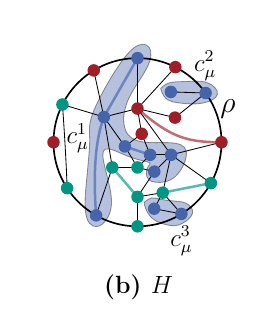}
    \label{fig:lemma-b}
    }
    \hfil
      \subfloat
  {
    \includegraphics[page=11,width=.23\textwidth]{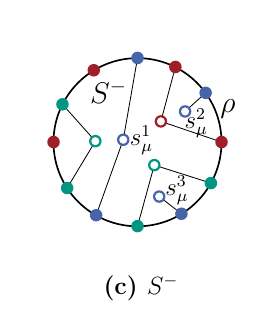}
    \label{fig:lemma-c}
    }
    \hfil
  \subfloat
  {
      \includegraphics[page=12,width=.23\textwidth]{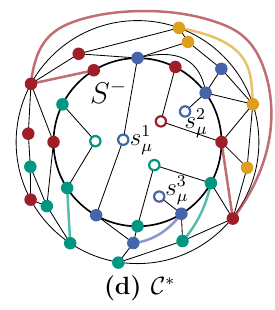}
    \label{fig:lemma-d}
    }
    \caption{
    (a) Super c-graph $\calC{'}$ of $\calC{}$. (b) Each component of the blue cluster $\mu$ in~$H$ lies inside a simple closed region. (c) Cycle-star $S^-$ corresponding to $H$. (d) The c-connected c-planar c-graph $\calC{*}$ obtained by {\em replacing} $H$ with $S^-$ in $\calC{'}$.}
  \label{fig:replacement}
\end{figure}

Let $\calC{-}_\Delta(R^-,{\cal J}^-)$ ($\calC{+}_\Delta(R^+,{\cal J}^+)$) be a c-graph obtained by augmenting the \mbox{c-graph} $\calC{-}(S^-,{\cal K}^-)$ ($\calC{+}(S^+,{\cal K}^+)$) of Lemma~\ref{le:split} by introducing new vertices, each belonging to a distinct cluster, and by adding edges only between the vertices in $V(S^-)$ ($V(S^+)$) and the newly introduced vertices in such a way that cycle $\rho$ bounds a face of $R^-$ ($R^+$) and all the other faces of $R^-$ ($R^+$) are triangles. From the construction of Lemma~\ref{le:split}, we also have the following useful technical remark.

\begin{restatable}{remarkx}{RemarkSplit}\label{rm:split}
The c-graph obtained by merging $\calC{-}_\Delta(R^-,{\cal J}^-)$ and $\calC{+}_\rho(G^+,\calT{+})$ (by merging $\calC{+}_\Delta(R^+,{\cal J}^+)$ and $\calC{-}_\rho(G^-,\calT{-})$) is c-planar.
\end{restatable} 

We now describe a divide-and-conquer algorithm based on Lemma~\ref{le:split}, called {\sc TestCP}, that tests the c-planarity of a $2$-connected c-graph $\calC{}(G,\calT{})$ and returns a super c-graph $\calC{*}(G^*,\calT{})$ of $\calC{}$ satisfying Condition~(\ref{con:3}) of Theorem~\ref{th:characterization}, if $\calC{}$ is c-planar. See Fig.~\ref{fig:algo} for illustrations of the c-graphs constructed during the execution of the algorithm.

We first give an intuition on the role of cycle-stars in Algorithm {\sc TestCP}.

Let $\calC{}(G,\calT{})$ be a c-planar c-graph and let $\rho$ be a cycle separator of $G$. 
By Lemma~\ref{le:split}, for each c-connected c-planar super c-graph $\calC{'}$ of $\calC{}$, we can injectively map the super c-graph $I^-$ of $\calC{-}_\rho(G^-,\calT{-})$, composed of the vertices of $G^-$ and of the edges in the interior and along the boundary of $\rho$ in $\calC{'}$, with a cycle-star~$S^-$ whose universal cycle is $\rho$. This is due to the fact that there exists a \emph{one-to-one correspondence} between the connected components of $I^-$ induced by the vertices of each cluster in $\calT{-}$ and the star vertices of $S^-$.
Similar considerations hold for the super c-graph $I^+$ of $\calC{+}_\rho(G^+,\calT{+})$.
Although the c-planarity of $\calC{+}_\rho$ and $\calC{-}_\rho$ is necessary for the c-planarity of $\calC{}$, it is not a sufficient condition, as the connectivity of clusters inside $\rho$ in $I^-$ (\emph{internal cluster-connectivity}) and the connectivity of clusters outside $\rho$ in $I^+$ (\emph{external cluster-connectivity}) must also together determine the c-connectedness of $\calC{'}$.
The role of cycle-stars $S^-$ and $S^+$ in the algorithm presented in this section is exactly that of concisely \emph{representing} the internal cluster-connectivity of $I^-$ and the external cluster-connectivity of $I^+$, respectively, to devise a divide-and-conquer approach \mbox{to test the c-planarity of~$\calC{}$.}

\paragraph{Outline of the algorithm}.
We  overview  the main steps of our algorithm below.
\begin{itemize}
\item If $n = O(\ell)$, we test $c$-planarity directly, as a base case for the divide-and-conquer recursion.
Otherwise, we construct a cycle-separator $\rho$ of $G$ and test whether $\rho$ is a cluster-separator.
If so, $\calC{}$ cannot be c-planar (Observation~\ref{obs:safe-cycle}), and we halt the search.
\item We generate all cycle-stars  $S^-_i$ with universal cycle $\rho$. A cycle-star $S^-_i$ \mbox{represents} a potential connection pattern of clusters inside~$\rho$.
 For each cycle-star $S^-_i$ we apply  Procedure {\sc OuterCheck} to test whether this \mbox{pattern} could be augmented by additional connections outside $\rho$ to complete the \mbox{desired} cluster-connectivity.
That is, we test whether $\calC{+}_\rho$ admits a c-connected \mbox{c-planar} super c-graph whose internal cluster-connectivity is represented by~$S^-_i$.
To test this, we replace the subgraph $G^-$ of $G$ in $\calC{}$ with an internally-triangulated supergraph $R^-_i$ of $S^-_i$ to obtain a c-graph $\calC{+}$ and recursively test $\calC{+}$ for c-planarity. It is important to observe that, the triangulation step prevents $\calC{+}$ from having saturating edges inside~$\rho$, thus enforcing exactly the same internal-cluster connectivity represented by $S^-_i$ (Remark~\ref{rm:split}). If $\calC{+}$ is c-planar, the procedure returns a c-connected c-planar super c-graph $\calC{+}_{con}$ of~$\calC{+}$. If no cycle-star \mbox{passes the test, $\calC{}$ is not} c-planar by Lemma~\ref{le:split}. 
We call all the cycle-stars that passed this test \emph{admissible}.
\item 
We then apply Procedure {\sc InnerCheck} to verify whether the internal-cluster connectivity represented by some admissible cycle-star $S^-_i$ can \emph{actually} be \mbox{realized} by a c-connected c-planar super c-graph of $\calC{}$. \mbox{For each admissible} cycle-star $S^-_i$, the procedure applies the construction of Lemma~\ref{le:split} to obtain a cycle-star $S^+_i$ representing the external cluster-connectivity of $\calC{+}_{con}$. Then, it tests whether $\calC{-}_\rho$ admits a c-connected c-planar super c-graph $\calC{-}_{con}$ whose external cluster-connectivity is represented by $S^+_i$. This is done similarly to Procedure \mbox{\sc OuterCheck}, by triangulating the exterior of $\rho$ and recursively testing $c$-planarity of a smaller graph. If Procedure {\sc InnerCheck} succeeds for any admissible cycle-star $S^-_i$, then we can merge the subgraphs of $\calC{-}_{con}$ and of $\calC{+}_{con}$ induced by the vertices inside and outside $\rho$, respectively, to obtain a c-connected c-planar super c-graph of $\calC{}$, and we halt the search with a successful outcome. It might be the case that $\calC{-}_{con}$ has a different internal-cluster connectivity than that represented by  $S^-_i$, but this is not a problem, because the different cluster connectivity (which necessarily corresponds to a different admissible cycle-star) still provides a c-planar drawing \mbox{of the whole graph.}
\item If no  admissible cycle-star passes Procedure {\sc InnerCheck}, $\calC{}$ is not c-planar.
\end{itemize}

It is crucial in this algorithm that $\rho$ be a cycle-separator. Because it is a cycle, no candidate saturating edges can connect vertices in the interior of $\rho$ to vertices in the exterior of $\rho$, as such vertices do not share any face. That is, the interaction between $G^-_{\rho}$ and $G^+_{\rho}$ only happens through vertices of $\rho$. This allows us to split the instance into smaller instances recursively along $\rho$ and
model the interaction via cycle-stars (by Lemma~\ref{le:split} and Remark~\ref{rm:split}) whose \mbox{universal cycle is~$\rho$.}

The complete listing of Algorithm {\sc TestCP} is provided in the next page.

\begin{figure}[tb!]
\hspace{-4mm}
    \centering
    \subfloat{
    \hspace{-2.2mm}
    \includegraphics[page=1,width=.2\textwidth]{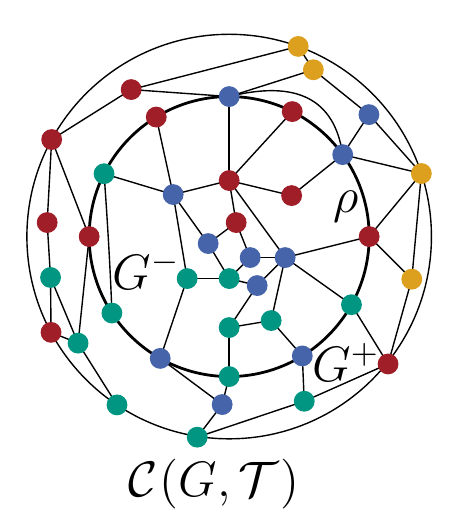}
    \label{fig:algo-a}
    \hspace{-2mm}
    }
  \subfloat{
    \includegraphics[page=2,width=.2\textwidth]{algorithm_f1}
    \label{fig:algo-b}
    \hspace{-2mm}
    }
  \subfloat{
      \includegraphics[page=3,width=.2\textwidth]{algorithm_f1}
    \label{fig:algo-c}
    \hspace{-2mm}
    }
  \subfloat{
    \includegraphics[page=4,width=.2\textwidth]{algorithm_f1}
    \label{fig:algo-d}
    \hspace{-3mm}
    }
  \subfloat{
    \includegraphics[page=5,width=.2\textwidth]{algorithm_f1}
    \label{fig:algo-e}
    }\\
  \subfloat{
    \includegraphics[page=6,width=.2\textwidth]{algorithm_f1}
    \label{fig:algo-b}
    }
  \subfloat{
    \includegraphics[page=7,width=.2\textwidth]{algorithm_f1}
    \label{fig:algo-b}
    }
  \subfloat{
      \includegraphics[page=8,width=.2\textwidth]{algorithm_f1}
    \label{fig:algo-c}
    }
  \subfloat{
    \includegraphics[page=9,width=.2\textwidth]{algorithm_f1}
    \label{fig:algo-d}
    }
    \caption{Illustrations of all of the c-graphs constructed by Algorithm {\sc TestCP}.
  }
  \label{fig:algo}
\end{figure}

\footnotetext[1]{The merging operations are well defined as cycle $\rho$ bounds  the outer face of $R^-_i$ and an inner face of $G^+$, as well as an inner face of $R^+_i$ and the outer face of $G^-$.}
\footnotetext[2]{As $\calC{+}(G^+_i,\calT{+}_i)$ and $\calC{-}(G^-_i,\calT{-}_i)$ are $2$-connected,\hspace{.1943em}{\sc TestCP} can be recursively applied.}

\paragraph{Base Case of the algorithm.}
The base case occurs when $\calC{+}(G^+_i, \calT{+}_i)$ and $\calC{-}(G^-_i, \calT{-}_i)$ are no longer smaller than $\calC{}(G, \calT{})$.

Observe that, we obtained $G^+_i$ ($G^-_i$) by merging $G^+$ ($G^-$) and $R^-_i$ ($R^+_i$) along cycle $\rho$, which has size $s(n)$. The size of $G^+$ and $G^-$ is bounded by \mbox{$\frac{2n}{3} + s(n)$}, while the size of $R^-_i$ and $R^+_i$ is bounded by $3 s(n)$. 
Therefore, since cycle $\rho$ is shared by all the mentioned graphs by construction,
we have that the size of $G^+_i$ and $G^-_i$ is at most $\frac{2n}{3}+3s(n)$. Thus, with \mbox{$s(n) \leq 2\sqrt{\ell n}$~\cite{m-fsscs2pg-86}}, we can set the base case of Algorithm~{\sc TestCP} when $n \leq \frac{2n}{3}+6\sqrt{\ell n}$, that is, $n \leq 324\ell$.

\paragraph{Correctness of the algorithm.}
We show that, given a $2$-connected c-graph $\calC{}(G,\calT{})$, Algorithm {\sc TestCP} returns {\tt YES}, which happens when both procedures {\sc OuterCheck} and {\sc InnerCheck} succeed, if and only if $\calC{}(G,\calT{})$ is c-planar.

($\Rightarrow$) Suppose that {\sc OuterCheck} and {\sc InnerCheck} succeed for a cycle-star $S^-_\omega \in {\cal S}$ constructed at step~\ref{alg:generateS}. We show \mbox{that $\calC{}(G,\calT{})$ is~c-planar.}
Consider the \mbox{c-graph} $\calC{*}(G^*,\calT{})$ constructed at step~\ref{algo:cPlanarCons} from $\calC{-}_{con}(H^-_\omega,\calT{-}_\omega)$ and $\calC{+}_{con}(H^+_\omega,\calT{+}_\omega)$.
The proof of this direction follows by the next claim about $\calC{*}$ and from Theorem~\ref{th:characterization}.

\newgeometry{
   letterpaper,         
   textwidth=12.2cm,  
   textheight=19.3cm, 
   hratio=1:1,      
   vratio=4:5,      
}

\begin{algorithm}
\caption*{\bf {\sc Algorithm TestCP(c-graph $\calC{}(G, \calT{})$)}}

\smallskip
\noindent\underline{\sc {Base case}}

\smallskip
If $|V(G)| = O(\ell)$, then we can test {\sc C-Planarity} for $\calC{}(G,\calT{})$ in $O(1)$ time when $\ell$ is a \mbox{constant}, by performing a brute force search to find a subset $E'$ of the candidate saturating edges of $\calC{}$ such that
\mbox{c-graph $\calC{'}(G\cup E', \calT{})$ satisfies Condition~(\ref{con:3}) of Theorem~\ref{th:characterization}.}

\smallskip
\noindent\underline{\sc {Recursive step}}
\begin{enumerate}
	\item \label{alg:selectRho} Select a cycle separator $\rho$ of $G$. If $\rho$ is a cluster-separator, then \Return{\tt NO}; otherwise, construct c-graphs $\calC{+}_{\rho}(G^+,\calT{+})$ and $\calC{-}_{\rho}(G^-,\calT{-})$ as defined in Lemma~\ref{le:split}.
	\item {\sc OuterCheck}
	\begin{enumerate}
    \item \label{alg:generateS} Construct the set $\cal S$ of all cycle-stars such that, for every $S \in {\cal S}$, it holds that (i) $\rho$ is the universal cycle of $S$, (ii) $\rho$ bounds the outer face of $S$, and (iii) every star vertex of $S$ is incident only to vertices of $\rho$ belonging to the same cluster.
		\item \label{alg:forS} For each cycle-star $S^-_i \in {\cal S}$:
		\begin{enumerate} 
			\item \label{alg:innerStar1} Construct a c-graph $\calC{-}(S^-_i,{\cal K}^-_i)$ as follows. First, initialize ${\cal K}^-_i$ to the subtree of $\calT{}$ whose leaves are the vertices of $S^-_i$. Then, for each star vertex $v
      $ of $S^-_i$, assign $v$ to the cluster $\mu \in {\cal K}^-_i$ to which all its neighbours belong.
			\item \label{alg:innerStar2} Augment $\calC{-}(S^-_i,{\cal K}^-_i)$ to an internally triangulated c-graph $\calC{-}_\Delta(R^-_i,{\cal J}^-_i)$ by introducing new vertices, each belonging to a distinct cluster, and by adding edges only between vertices in $V(S^-_i)$ and the newly introduced vertices (that is, no two non-adjacent vertices in $S^-_i$ are adjacent in $R^-_i$). 
			\item \label{alg:innerStar3} Merge $\calC{-}_\Delta(R^-_i,{\cal J}^-_i)$ and $\calC{+}_{\rho}(G^+,\calT{+})$ to obtain a c-graph $\calC{+}(G^+_i,\calT{+}_i)$\red{\footnotemark[1]}.
			\item \label{alg:recursionOut} Run {\sc TestCP}($\calC{+}(G^+_i,\calT{+}_i)$) to test whether $\calC{+}(G^+_i,\calT{+}_i)$ is c-planar\red{\footnotemark[2]}. 
		\end{enumerate}	
		\item If no c-graph $\calC{+}(G^+_i,\calT{+}_i)$ is c-planar, then \Return{\tt NO}; otherwise, initialize ${\cal S}'$ as the
            set of {\em admissible} cycle-stars, i.e., the cycle-stars in $\cal S$ whose corresponding c-graph $\calC{+}(G^+_i,\calT{+}_i)$ is c-planar.
	\end{enumerate}
	\item {\sc InnerCheck}
	\begin{enumerate} 
		\item For each admissible cycle-star $S^-_i \in {\cal S}'$:
		\begin{enumerate} 
		\item \label{algo:outerStar1} Let $\calC{+}_{con}(H^+_i,\calT{+}_i)$ be the c-planar c-connected super c-graph of $\calC{+}$ returned by {\sc TestCP}($\calC{+}(G^+_i,\calT{+}_i)$) (step~\ref{alg:recursionOut}). 
		Apply the construction of Lemma~\ref{le:split} to c-graph $\calC{+}_{con}(H^+_i,\calT{+}_i)$ and cycle $\rho$ to obtain a c-graph $\calC{+}(S^+_i,{\cal K}^+_i)$ satisfying Properties~(\ref{prop:1}) and~(\ref{prop:2}) of the lemma. 
		\item \label{algo:outerStar2} Augment $\calC{+}(S^+_i,{\cal K}^+_i)$ to a c-graph $\calC{+}_\Delta(R^+_i,{\cal J}^+_i)$ 
		by introducing new vertices, each belonging to a distinct cluster, and by adding edges only between the vertices in $V(S^+_i)$ and the newly introduced vertices in such a way that cycle $\rho$ bounds an inner face of $R^+_i$ and all the other faces of $R^+_i$ are triangles.
		\item \label{algo:outerStar3} Merge $\calC{+}_\Delta(R^+_i,{\cal J}^+_i)$ and $\calC{-}_{\rho}(G^-,\calT{-})$ to obtain a c-graph $\calC{-}(G^-_i,\calT{-}_i)$\red{\footnotemark[1]}.
		\item \label{alg:recursionIn} Run {\sc TestCP}($\calC{-}(G^-_i,\calT{-}_i)$) to test whether $\calC{-}(G^-_i,\calT{-}_i)$ is c-planar\red{\footnotemark[2]}. 
		\item \label{algo:cPlanarCons} If {\sc TestCP}($\calC{-}(G^-_i,\calT{-}_i)$) returns {\tt YES}, then construct a c-planar c-connected super c-graph $C^*(G^*,\calT{})$ of $\calC{}(G, \calT{})$ as follows. 
		Let $\calC{-}_{con}(H^-_i,\calT{-}_i)$ be the c-planar c-connected c-graph returned by {\sc TestCP}($\calC{-}(G^-_i,\calT{-}_i)$). Remove all the vertices and edges of $H^-_i$ in the exterior of cycle $\rho$, thus obtaining a new c-graph $\calC{}_{in}(G_{in}, \calT{}_{in})$ in which cycle $\rho$ bounds the outer face. Similarly, remove all the vertices and edges of $H^+_i$ in the interior of cycle $\rho$, thus obtaining a new c-graph $\calC{}_{out}(G_{out}, \calT{}_{out})$ in which cycle $\rho$ bounds an inner face.
		Finally, merge $\calC{}_{in}$ and $\calC{}_{out}$ to obtain c-graph $\calC{*}(G^*,\calT{})$ and \mbox{\Return{\tt YES}} along with c-graph $\calC{*}(G^*,\calT{})$.
		\end{enumerate} 
	\end{enumerate} 
	\item \Return{\tt NO} if no c-graph $\calC{-}(G^-_i,\calT{-}_i)$, constructed at step~\ref{algo:outerStar3}, is c-planar.
\end{enumerate}
\end{algorithm}
\newgeometry{left=2cm,right=2cm,top=2cm,bottom=2cm} 


\begin{restatable}{claimx}{directdirection}\label{claim:directdirection}
C-graph $\calC{*}(G^*,\calT{})$ is a 
c-planar c-connected super c-graph of $\calC{}(G,\calT{})$.
\end{restatable}

\begin{proof}
Graphs $G_{in}$ and $G_{out}$ are planar, as they are subgraphs of $H^-_\omega$ and $H^+_\omega$, respectively (step~\ref{algo:cPlanarCons}). By construction, cycle $\rho$ bounds an inner face of $G_{out}$ and the outer face of $G_{in}$. Therefore $G^*$, obtained by merging $G_{in}$ and $G_{out}$, is planar. Also, observe that, $G_{in}$ and $G_{out}$ are supergraphs of $G^-$ and $G^+$, respectively, therefore graph $G^*$ is a super graph of $G$.

We now show that $\calC{*}$ is c-connected, that is, for each cluster $\mu \in \calT{}$, graph $G^*(\mu)$ is connected.

First, let $\mu$ be a cluster in $\calT{}$ such that $V(\mu)$ lies in the interior of $\rho$ in $G$. Since $\calC{-}_{con}(H^-_\omega,\calT{-}_\omega)$ is c-connected, we have that $H^-_\omega(\mu)$ is connected. Also, $V(\mu)$ lie in the interior of $\rho$ in $H^-_\omega$. By construction, $G_{in}$ contains all the vertices and the edges in the interior of $\rho$, therefore  we also have that $G_{in}(\mu)$ is connected. Hence, $G^*(\mu)$ is connected. The proof that graph $G^*(\mu)$ is connected, for each cluster $\mu$ in $\calT{}$ such that $V(\mu)$ lies in the exterior of $\rho$ in $G$, is analogous.

Then, let $\mu$ be a cluster such that $V(\mu) \cap V(\rho) \neq \emptyset$. 
Clearly, if $V(\mu) \subseteq V(\rho)$, then $G^*(\mu)$ is connected since both $G_{in}(\mu)$ and $G_{out}(\mu)$ are connected.
Otherwise, the following three cases are possible: either $G_{in}(\mu)$ is disconnected, or $G_{out}(\mu)$ is disconnected, or both $G_{in}(\mu)$ and $G_{out}(\mu)$ are disconnected. 

We show that all the vertices in $G_{in}(\mu)$ and in $G_{out}(\mu)$ are connected in~$G^*(\mu)$.


We first prove that all the vertices in $G_{in}(\mu)$ are connected in~$G^*(\mu)$.

Consider two connected components  $c'$ and $c''$ of $G_{in}(\mu)$. 
Observe that, by construction, c-graph $\calC{-}_{con}(H^-_\omega,\calT{-}_\omega)$ (step~\ref{algo:cPlanarCons}) is a merge of $\calC{}_{in}(G_{in}, \calT{}_{in})$ and of $\calC{+}_\Delta(R^+_\omega,{\cal J}^+_\omega)$.
Since $\calC{-}_{con}$ is c-connected and since $R^+_\omega$ is an augmentation of cycle-star $S^+_\omega$ such that edges in $E(R^+_\omega) \setminus E(S^+_\omega)$ do not have endpoints in the same cluster, the c-graph $\calC{\#}(G^\#,\calT{\#})$ obtained by merging $\calC{}_{in}$ and $\calC{+}(S^+_\omega,{\cal K}^+_\omega)$ is also c-connected.
Since $\calC{\#}$ is c-connected,
the vertices of $c'$ and $c''$ are connected via star vertices of $S^+_\omega$ and vertices of $G_{in}$ belonging to cluster $\mu$ in~$G^\#(\mu)$.
Observe that, by construction, c-graph $\calC{+}_{con}(H^+_\omega,\calT{+}_\omega)$  is a merge of $\calC{}_{out}(G_{out}, \calT{}_{out})$ and of $\calC{-}_\Delta(R_\omega,{\cal J}_\omega)$. Further, $S^+_\omega$ has been obtained by applying the construction of Lemma~\ref{le:split} to c-graph $\calC{+}_{con}(H^+_\omega,\calT{+}_\omega)$ (step~\ref{algo:outerStar1}) and cycle $\rho$.  Therefore, each connected component of $\mu$ in $G_{out}$ corresponds to a star vertex of $S^+_\omega$. Hence, we have that the vertices of $c'$ and $c''$ are also connected in $G^*$ via vertices of $G_{out}$ and $G_{in}$ belonging to cluster $\mu$.

Now, we prove that all the vertices in $G_{out}(\mu)$ are connected in~$G^*(\mu)$.

Consider two connected components  $c'$ and $c''$ of $G_{out}(\mu)$.
Observe that, as shown above, each connected component of $\mu$ in $G_{out}$ corresponds to a star vertex of $S^+_\omega$. 
Recall that $\calC{\#}$ is c-connected. Therefore, the star vertices of $S^+_\omega$ corresponding to $c'$ and $c''$ are connected via other star vertices of $S^+_\omega$ and vertices of $G_{in}$ belonging to cluster $\mu$ in~$G^\#(\mu)$. Hence, the vertices of $c'$ and $c''$ are also connected in $G^*$ via vertices of $G_{out}$ belonging to connected components of $\mu$ corresponding to star vertices of $S^+_\omega$ and vertices of $G_{in}$ belonging to cluster $\mu$ in~$G^*(\mu)$.
\end{proof}

%





($\Leftarrow$) Suppose that $\calC{}(G,\calT{})$ is c-planar. We show that Procedure  {\sc OuterCheck} and {\sc InnerCheck} succeed. Since $\calC{}(G,\calT{})$ is c-planar, there exists a super c-graph $\calC{*}(G^*,\calT{})$ of $\calC{}$ such that $G^*$ is planar and $\calC{*}$ is c-connected, by Theorem~\ref{th:characterization}. By using the construction of Lemma~\ref{le:split} on c-graph $\calC{*}$,  we can obtain a cycle-star $S^-$ whose universal cycle is $\rho$ that represents the connectivity of clusters inside $\rho$ in $\calC{*}$.
The proof of this direction follows from the next claim.
\begin{restatable}{claimx}{reversedirection}\label{cl:reversedirection}
Procedures  {\sc OuterCheck} and~{\sc InnerCheck} succeed for $S^-_i=S^-$.
\end{restatable}

\begin{proof}
Procedure  {\sc OuterCheck} succeeds if, for a cycle separator $\rho$ of $G$ selected at step~\ref{alg:selectRho} of the algorithm, there exists a cycle-star $S^-_i$ whose universal cycle is $\rho$ such that the corresponding c-graph $\calC{+}(G^+_i,\calT{+}_i)$, constructed at steps~\ref{alg:innerStar1}, \ref{alg:innerStar2}, and~\ref{alg:innerStar3} of the algorithm, is c-planar. 
Recall that, cycle-star $S^-$ has the following properties: 
\begin{inparaenum} 
\item Cycle $\rho$ is the universal cycle of $S^-$ and bounds the outer face of $S^-$, and
\item for each star vertex $v$ of $S^-$, the neighbours of $v$ belong to the same cluster $\mu \in {\cal K}^-$vertex $v$ belongs to.
\end{inparaenum}
Since, steps~\ref{alg:generateS} and~\ref{alg:innerStar1} construct all c-graphs $\calC{-}(S^-_i,{\cal K}^-_i)$ with the above properties, when $S^-_i=S^-$ we are guaranteed to compute c-graph $\calC{-}(S^-,{\cal K}^-)$. 
First, observe that
the c-graph obtained by merging c-graphs $\calC{-}(S^-,{\cal K}^-)$ and $\calC{+}_\rho(G^+,\calT{+})$ is c-planar, since $S^-$ 
has been obtained by applying the construction of Lemma~\ref{le:split} to a super c-connected c-planar c-graph $\calC{*}$ of $\calC{}$. This 
 together with Remark~\ref{rm:split} imply that c-graph $\calC{+}(G^+_i,\calT{+}_i)$ is c-planar. Thus, the invocation of {\sc TestCP} on $\calC{+}(G^+_i,\calT{+}_i)$ at step~\ref{alg:recursionOut} will return {\tt YES}. Hence, Procedure {\sc OuterCheck} succeeds.

Procedure  {\sc InnerCheck} succeeds if, there exists a c-graph $\calC{-}(G^-_i,\calT{-}_i)$, constructed at steps~\ref{algo:outerStar1}, \ref{algo:outerStar2}, and~\ref{algo:outerStar3} of the algorithm, that is c-planar. 
By Theorem~\ref{th:characterization}, a c-graph $\calC{-}(G^-_i,\calT{-}_i)$ is c-planar if and only if there exists a super c-graph $\calC{'}(G',\calT{-}_i)$ of $\calC{-}(G^-_i,\calT{-}_i)$ such that $G'$ is planar and $\calC{'}$ is c-connected. 
As Procedure  {\sc OuterCheck} succeeds, the c-graph $\calC{+}(G^+_i,\calT{+}_i)$ corresponding to $S^-$ is c-planar. Therefore, Procedure  {\sc OuterCheck} provides us with a c-planar c-connected c-graph $\calC{+}_{con}(H^+_i,\calT{+}_i)$ (see steps~\ref{alg:recursionOut} and~\ref{algo:outerStar1}) that is a super c-graph of $\calC{+}(G^+_i,\calT{+}_i)$. 
Consider the c-graph $\calC{+}(S^+_i,{\cal K}^+_i)$ constructed at step~\ref{algo:outerStar1} by applying the construction of Lemma~\ref{le:split} to $\calC{+}_{con}$. 
Observe that, the c-graph obtained by merging $\calC{+}(S^+_i,{\cal K}^+_i)$ and $\calC{-}_\Delta(R^-_i,{\cal J}^-_i)$ is a c-connected c-planar c-graph. 
This is due to the fact that, since $R^-_i$ is internally triangulated, there exists no edge in the interior of $\rho$ in $H^+_i$ that belongs to $H^+_i$ and does not belong to $R^-_i$, that is, no candidate saturating edges connect two vertices in the interior of $\rho$ in $\calC{+}_{con}$.
Since $S^+_i \subseteq R^+_i$, we also have that the c-graph obtained by merging $\calC{+}_\Delta(R^+_i,{\cal J}^+_i)$ (constructed at step~\ref{algo:outerStar2}) and $\calC{-}_\Delta(R^-_i,{\cal J}^-_i)$ is a c-connected c-planar c-graph.
 Also, since each of the vertices added to obtain $R^-_i$ from $S^-$ belongs to a different cluster and since the edges added to internally triangulate $S^-$ do not connect vertices of the same cluster, we have that the c-graph obtained by merging $\calC{+}_\Delta(R^+_i,{\cal J}^+_i)$ and $\calC{-}(S^-,{\cal K}^-)$ is also a c-connected c-planar c-graph.

Let $\cal A$  be the subgraph of $G^*$ induced by the edges in the interior and on the boundary of $\rho$ in $\calC{*}$.
Since $S^-$ exactly represents the cluster connectivity of $\cal A$, 
 the c-graph obtained by merging $\calC{+}_\Delta(R^+_i,{\cal J}^+_i)$ and $\cal A$ is also a c-connected c-planar c-graph. The fact that, such a c-graph is a super c-graph of $\calC{-}(G^-_i, \calT{-}_i)$ shows that $\calC{-}(G^-_i,\calT{-}_i)$ is c-planar. Hence, Procedure {\sc InnerCheck} succeeds. 
 \end{proof}


\medskip

We are finally ready to present the main result of the section.

\begin{restatable}{theorem}{subexponential}\label{th:subexponential}
The {\sc C-Planarity} problem can be solved in $2^{O(\sqrt{\ell n}\cdot{\log{n}})}$ time for $n$-vertex c-graphs with maximum face size $\ell$.
\end{restatable}

\begin{proof}
Given an $n$-vertex c-graph $\calC{}(G,\calT{})$ with maximum face size $\ell$, by Lemma~\ref{lem:3connected}, we can construct in linear time a $2$-connected, in fact $3$-connected, c-graph $\calC{'}$ equivalent to $\calC{}$. Therefore, we can apply Algorithm {\sc TestCP} to $\calC{'}$ to determine whether $\calC{}$ is c-planar.

We now argue about the running time.

Since $G'$ is $2$-connected and since, by Lemma~\ref{lem:3connected},
$|V(G')| = O(|V(G)|)$ and the maximum face size $\ell'$ of $G'$ is $O(\ell)$, we can construct a cycle separator $\rho$ of $G$ of size $s(n)=O(\sqrt{\ell n})$ that separates the vertices in the interior of $\rho$ from the vertices in the exterior of $\rho$ in such a way that both such sets contain at most $\frac{2n}{3}$ vertices~\cite{m-fsscs2pg-86}.
Also, since all cycle-stars whose universal cycle is $\rho$ have size $O(s(n))$ and
the augmentations at steps \ref{alg:innerStar2} and~\ref{algo:outerStar2} can be done by introducing at most $s(n)$ new vertices, graphs $G_i^+$ (step~\ref{alg:recursionOut}) and $G^-_i$ (step~\ref{alg:recursionIn}) have $O(\frac{2n}{3} + O(s(n)))$ size. Further, by construction,  $G_i^-$ and $G_i^+$ are $2$-connected and their maximum face size is $\ell'$; thus, the cycle separators of $G_i^-$ and $G_i^+$ have size bounded by $s(|V(G_i^-)|)$ and by $s(|V(G_i^+)|)$, respectively.

Moreover, observe that each cycle-star $S^-_i \in \cal S$ satisfying the properties described at step~\ref{alg:generateS} can be constructed in $O(s(n))$ time. Also, each cycle-star $S^-_i$ is in one-to-one correspondence with a non-crossing partition of a set containing $s(n)$ elements. This is due to the fact that each vertex of $\rho$ is incident to at most a star vertex of $S^-_i$ and that, by the planarity of $S^-_i$, the neighbours of any two star vertices do not alternate along $\rho$. The number of all such partitions is expressed by the Catalan number $C_{s(n)} \leq 4^{s(n)}$. 

The non-recursive running time $f(n)$ is bounded by the time taken by steps~\ref{alg:selectRho} and~\ref{algo:outerStar1}, that is, $O(n)$ time. In fact, the cycle-separator of an $n$-vertex graph can constructed in $O(n)$ time~\cite{m-fsscs2pg-86}. Testing whether a cycle is a cluster-separator can be done by performing a visit of the graph to detect if there exist a cluster whose vertices lie inside and outside of $\rho$, but not along $\rho$; this can clearly be done in $O(n)$ time. Finally, applying the construction of Lemma~\ref{le:split} to obtain a cycle-star only requires finding the connected components of each cluster inside (or outside) $\rho$ and their respective connections to cycle $\rho$, which can be done in $O(n)$ time by performing a DFS-visit of $G^-$ (or $G^+$).

By the above arguments, the running time of Algorithm {\sc TestCP} is expressed by by the following recurrence:

\begin{equation}\label{eq:runnintime}
T(n)=2 C_{s(n)} \bigg(T\Big(\frac{2n}{3} + O\big(s(n)\big)\Big) + f(n)\bigg)
\end{equation} 

Since equation~(\ref{eq:runnintime}) solves to $T(n)= 2^{O(\sqrt{\ell n}\cdot{\log{n}})}$ for $s(n) = O(\sqrt{\ell n})$, $C_{s(n)} = 4^{s(n)}$, $f(n)=O(n)$, the statement follows.
%
\end{proof}

In the next section, we show how to adapt algorithm {\sc TestCP} to obtain an $\mathsf{XP}$ algorithm with parameter $h$ for generalized $h$-simply nested graphs, which extend simply-nested graphs with bounded face size.

\input{h-generalized}

%% file: h-generalized.tex
\subsection{Generalized $h$-Simply-Nested Graphs}\label{se:simplynested}


\begin{wrapfigure}[12]{R}{.31\textwidth}
    \centering
    \vspace{-9.4mm}
    \includegraphics[page=1, width=.3\textwidth]{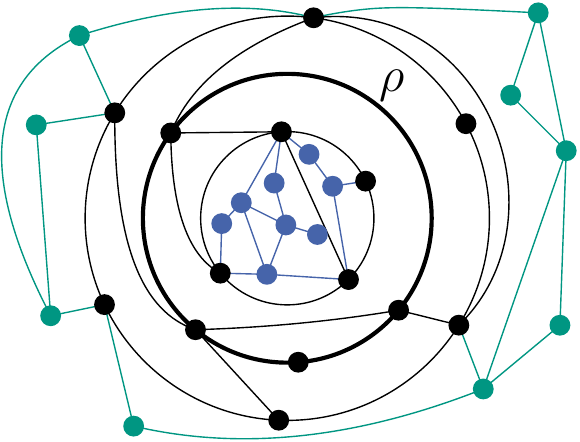}
    \caption{A generalized \mbox{$6$-simply-nested} graph. 
    }
    \label{fig:hgeneralized}
\end{wrapfigure}
A plane graph is \mbox{\em $h$-simply-nested} if it consists of nested cycles of size at most $h$ and of edges only connecting vertices of the same cycle or vertices of adjacent cycles; refer to Fig.~\ref{fig:hgeneralized}. 
We extend the class of $h$-simply-nested graphs to the class of \mbox{\em generalized $h$-simply-nested graphs}, by allowing the inner-most cycle to contain a plane graph consisting of at most $2h$ vertices in its interior and the outer-most cycle to 
contain a plane graph consisting of at most $2h$ vertices in its exterior. 
See~\cite{Cimikowski-90} for a related graph class, in which the vertices in the interior of the inner-most cycle can only form a tree, there exist no other vertices in the exterior of the outer-most cycle, and chords are not allowed for the remaining cycles.

Let $G$ be a generalized $h$-simply-nested plane graph with $n> 5h$ vertices. We have the following simple observation about the structure of $G$; refer to Fig.~\ref{fig:hgeneralized}.

\begin{observation}\label{obs:generalize-structure}
Graph $G$ contains a cycle $\rho$ with $|V(\rho)| \leq h$ that separates $G$ into two generalized $h$-simply-nested graphs $G^+$ and $G^-$ with $|V(G^+)| \leq \frac{n}{2}$ and $|V(G^-)| \leq \frac{n}{2}$ such that
$G^+$ ($G^-$) does not contain any vertex in the exterior (interior) of its outer-most cycle (inner-most cycle). Further, such a cycle can be computed in $O(n)$ time.
\end{observation}

By Observation~\ref{obs:generalize-structure}, we can use a cycle separator of size at most $h$ in \mbox{Algorithm} {\sc TestCP} to test the c-planarity of a c-graph whose underlying graph is a generalized $h$-simply-nested plane graph $G$ (instead of a cycle separator of size $O(\sqrt{\ell n})$, where $\ell$ is the maximum face size of $G$). Observe that, graphs $G_i^+$ and $G_i^-$ obtained at steps~\ref{alg:innerStar3} and~\ref{alg:innerStar3} of the algorithm, respectively, also belong to the family of generalized $h$-simply-nested plane graphs. Therefore, Observation~\ref{obs:generalize-structure} also holds for such graphs. Altogether, we obtain the following recurrence relation for the running-time:

\begin{equation}\label{eq:runnintime-generalized}
T(n)=2 C_{h} \bigg(T\Big(\frac{n}{2} + O(h)\Big) + O(n)\bigg)
\end{equation}

Equation~(\ref{eq:runnintime-generalized}) immediately implies the following theorem. 

\begin{theorem}
The {\sc C-Planarity} problem can be solved in $n^{O(h)}$ time for 
\mbox{$n$-vertex} c-graphs whose underlying graph is a generalized $h$-simply-nested~graph.
\end{theorem}

%% file: courcelle.tex

\section{An MSO$_2$ formulation for C-Planarity}\label{se:courcelle}

In this section, we show that the property of a c-graph of admitting a c-planar drawing can be expressed 
in extended monadic second-order ($\msotwo$) logic. We use this result and the fact that graph properties definable in $\msotwo$ logic can be verified in linear time on graphs of bounded treewidth, by Courcelle's Theorem~\cite{Courcelle90}, to build an FPT algorithm for testing \mbox{the c-planarity of embedded flat c-graphs.}  

{\em First-order graph logic} deals with formul{ae} whose variables represent the vertices and edges of a graph. 
{\em Second-order graph logic} also allows quantification over $k$-ary relations defined on the vertices and edges.
{\em MSO$_2$ logic} only allows quantification over elements and {\em unary relations}, that is, sets of vertices and edges.
Given a graph $G$ and an $\msotwo$ formula $\phi$, we say that $G$ {\em models } $\phi$, denoted by $G \models \phi$, if the logic statement expressed by $\phi$ is satisfied by the vertices, edges, and sets of vertices and edges in~$G$. We will apply Courcelle's theorem not to the underlying graph $G$ of the clustered planarity instance, but to the supergraph $G^\diamond$ of $G$ that includes all the candidate saturating edges of $G$. This will allow us to quantify over sets of candidate saturating edges, but in exchange we must show that $G^\diamond$, and not just $G$, has low treewidth (Lemma~\ref{lem:3connected}). 

Let $H$ be a graph and let $E_1,E_2 \subseteq E(H)$. The following logic predicates can be expressed in $\msotwo$ logic (refer, e.g., to~\cite{BannisterE14,cfklmpps-pa-15} for their detailed \mbox{formulation):}
\begin{itemize}[$\diamond$]
\item $\planar_H(E_1,E_2):=$ the subgraph $(V(H), E_1 \cup E_2)$ of $H$ is planar, and
\item $\connected_H(U,E_1,E_2):=$ vertices in $U \subseteq V(H)$ are connected by edges in~$E_1 \cup E_2$.
\end{itemize}

Let $\calC{}(G,\calT{})$ be a c-graph and let $E^*$ be the set of all the candidate saturating edges of $\calC{}$. By Property(\ref{con:3}) of Theorem~\ref{th:characterization}, c-graph $\calC{}$ admits a c-planar drawing if and only if 
there exists a super c-graph $\calC{'}(G',\calT{})$ of $\calC{}$ such that $G'$ is planar and $\calC{'}$ is c-connected.
%
%
Testing Property(\ref{con:3}) amounts to determining the existence of a set $E^+ \subseteq E^*$ such that
\begin{inparaenum}[(i)]
\item the subgraph $G'$ of $G^\diamond$ obtained by adding the edges in $E^+$ to $G$ is planar and 
\item graph $G'(\mu)$ is connected, for each \mbox{cluster $\mu \in \calT{}$.}
\end{inparaenum}

We remark that in an $\msotwo$ formula it is possible to refer to given subsets of vertices or edges of a graph, provided that the elements of such subsets can be distinguished from the elements of other subsets by equipping them with labels from a constant finite set~\cite{ALS91}. Therefore, in our formulae we use the unquantified variables $V_i$ to denote the set of vertices belonging to cluster $\mu_i$, for each disconnected cluster $\mu_i \in \calT{}$, $E^*$ to denote the set consisting of all the candidate saturating edges of~$\calC{}$, and $E_G$ to denote $E(G)$.

Let $c$ be the number of disconnected clusters in $\calT{}$. We have the formula: 
%

$$
\cplanar_{{\cal C}(G,{\cal T})}
{\footnotesize
\equiv 
\exists (E^+ \subseteq E^*) \Big [\planar_{G^\diamond}(E_G,E^+) \wedge
\bigwedge^c_{i=1} \connected_{G^\diamond}(V_i,E_G,E^+)\Big ]}
$$

It is easy to see that formula $\cplanar_{{\cal C}(G,{\cal T)}}$ correctly expresses Condition(\ref{con:3}) of Theorem~\ref{th:characterization} only if $G$ admits a unique combinatorial embedding (up to a flip). In fact, if $G$ has more than one embedding, formula $\cplanar_{{\cal C}(G,{\cal T)}}$ might still be satisfiable after a change of the embedding, as formula $\planar_{G^\diamond}(E_G,E^+)$ models the planarity of an abstract graph rather than the planarity of a combinatorial embedding. We formalize this fact in the following lemma.

\begin{lemma}\label{le:courcelle}
Let $\calC{}(G,\calT{})$ be a c-graph such that $G$ has a unique combinatorial embedding and let $\calC{\diamond
}(G^\diamond,\calT{\diamond})$ be the c-graph obtained by augmenting $\calC{}$ with all its candidate saturating edges. 
Then, $\calC{}$ is c-planar \mbox{iff $G^\diamond \models \cplanar_{{\cal C}(G,{\cal T})}$.}
\end{lemma}

Since changes of embedding are not allowed in our context, as we aim at testing the c-planarity of a c-graph given a prescribed embedding, we combine Lemmata~\ref{lem:3connected} and~\ref{le:courcelle}, and then {invoke Courcelle's Theorem to obtain the following main result.}


\begin{theorem}\label{th:courcelle}
The {\sc C-Planarity} problem can be solved in $f(\overline{\emw},c) O(n)$ time for 
$n$-vertex c-graphs with $c$ disconnected clusters and whose \mbox{underlying graph} has embedded-width $\overline{\emw}$, \mbox{where $f$ is a computable function.}
\end{theorem}


\newcommand{\runningtimethfour}{
We now argue about the running time. 
By Lemma~\ref{lem:3connected}, c-graph $\calC{*}(G^*,\calT{*})$ can be constructed in $O(n)$ time.
Let $\kappa$ be the maximum face size of $G^*$.
The number of candidate saturating edges of $\calC{*}$ is $O(\kappa^2 n)$. By Lemma~\ref{lem:3connected},
$\kappa = O(\ell)$.
Hence, we can augment $\calC{*}(G^*,\calT{*})$ to obtain $\calC{\diamond}(G^\diamond,\calT{\diamond})$ in $O(\ell^2 n)$~time.

By Courcelle's theorem~\cite{Courcelle90}, it is possible to verify whether $G^\diamond \models \phi$ in $g(tw(G^\diamond),len(\phi))O(|V(G^\diamond)|+|E(G^\diamond)|)$ time, where $g$ is a computable function.
By Lemma~\ref{lem:3connected}, $|V(G^\diamond)|=|V(G^*)|=O(n)$ and $\tw(G^\diamond) = \overline{\emw}(G)$. Also, by the discussion above, $|E(G^\diamond)| = O(\ell^2 n)$ and, by definition of embedded-width, $\ell = O(\overline{\emw})$; thus, $|E(G^\diamond)| = O(\overline{\emw}^2 n)$.
Further, formula $\phi$ can be constructed in time proportional to its length $len(\phi)$, which is $O(c)$.
Therefore, the overall running time can be expressed as $f(\overline{\emw},c) O(n)$, where \mbox{$f$ is a computable function.}}

\begin{proof}
To test that $\calC{}(G,\calT{})$ admits a c-planar drawing with the given embedding we proceed as follows.
First, we apply Lemma~\ref{lem:3connected} to obtain a c-graph $\calC{*}(G^*,\calT{*})$ that is equivalent to $\calC{}(G,\calT{})$ such that $G^*$ is $3$-connected. 
Note that, the \mbox{$3$-connectivity} of $G^*$ implies that it has a unique combinatorial embedding (up to a flip)~\cite{Whitney33}. 
Then, we construct formula $\phi = \cplanar_{{\calC{*}}(G^*,{\calT{*})}}$ and the super c-graph $\calC{\diamond
}(G^\diamond,\calT{\diamond})$ of $\calC{*}$ obtained by augmenting $\calC{*}$ with all its candidate saturating edges. Finally, we use Courcelle's Theorem to test whether $G^\diamond \models \phi$. The correctness immediately follows \mbox{from Lemmata~\ref{lem:3connected} and~\ref{le:courcelle}.}
%
%
%

\runningtimethfour{}\end{proof}